\documentclass{article}

\usepackage{hyperref}
\usepackage{graphicx}%
\usepackage{multirow}%
\usepackage{amsmath,amssymb,amsfonts}%
\usepackage{amsthm}%
\usepackage{mathrsfs}%
\usepackage[title]{appendix}%
\usepackage{xcolor}%
\usepackage{textcomp}%
\usepackage{manyfoot}%
\usepackage{booktabs}%
\usepackage{algorithm}%
\usepackage{algorithmic}%
\usepackage{listings}%
\usepackage{amsthm,amssymb,latexsym,graphicx,amsmath,lineno}
\newtheorem{theorem}{Theorem}
\newtheorem{remark}[theorem]{Remark}
\newtheorem{lemma}[theorem]{Lemma}
\newtheorem{proposition}[theorem]{Proposition}

\newtheorem{example}[theorem]{Example}

\title{Experimental algorithms for the dualization problem}

\author{Mauro Mezzini
\and Fernando Cuartero Gomez
\and Jose Javier Paulet Gonzalez
\and Hernan Indibil de la Cruz Calvo
\and Vicente Pascual
\and Fernando L. Pelayo}

\begin{document}

\maketitle

\begin{abstract}
In this paper, we present experimental algorithms for solving the dualization problem. We present the results of extensive experimentation comparing the execution time of various algorithms.
\end{abstract}


\section{Introduction}\label{sec1}

The dualization problem was one of the most prominent open problem in computer science until recently \cite{DBLP:journals/qmi/MezziniGGCPP24} when its complexity was determined to be in $P$.

The dualization problem is formalized as follow. Given two Boolean functions $f:\{0,1\}^n \to \{0,1\}$ and $g:\{0,1\}^n \to \{0,1\}$ we say that $g \leq f$ if $g(x) \leq f(x)$  for all $x \in \{0, 1\}^n$. 
Given two Boolean vectors $v=(v_1, \dots, v_n)$ and $w=(w_1, \dots, w_n)$, we write $v \preceq w$ if $v_i \leq  w_i$ for all $i \in \{1,2,\dots, n\}$. 

A \emph{positive} (or elsewhere called \emph{monotone} \cite{doi:10.1137/S009753970240639X,10.1016/j.dam.2007.04.017,10.1016/S0166-218X(99)00099-2}) Boolean function satisfies the proposition that if $v \preceq w$ then $f(v) \leq f(w)$ \cite{DBLP:books/daglib/0028067}.

When $g \leq f$ we say that $g$ is an \emph{implicant} of $f$. An implicant $g$ of a function $f$ is called \emph{prime},  if there is no implicant $h\neq g$ of $f$ such that $g \leq h$.

A literal is a Boolean variable $x$ or its negation $\overline x$. It is known \cite{DBLP:books/daglib/0028067} that a positive Boolean function $f$ can be expressed by a disjunctive normal form (DNF) containing no negated literals. We will call it a \emph{positive} DNF expression of $f$. In the following, we will denote a positive DNF expression  of a positive Boolean function $f$ as
\begin{equation}
\varphi= \bigvee_{e \in F} \bigwedge_{i \in e} x_i
\end{equation} 
where $F$ is a set of subsets of $\{1,2, \dots, n\}$. For any $e \in F$ the implicant $\bigwedge_{i \in e} x_i$  of $f$ is also called \emph{term} of the DNF $\varphi$ and will be denoted simply by $e$. In the following we often identify $\varphi$ with its set of terms $F$. A positive DNF expression of a positive Boolean function 
it is said \emph{irredundant} if there is no $t \in F$ such that 
\[
\psi =\bigvee_{e \in F, e \neq t} \bigwedge_{i \in e} x_i
\]
is another positive DNF representation of $f$.   
It is well known \cite{DBLP:books/daglib/0028067} that a positive DNF which contains all and only the implicants of a positive Boolean function $f$ is unique and irredundant. We will call it \emph{positive irredundant DNF} (PIDNF) (elsewhere called \emph{prime DNF} \cite{doi:10.1137/S009753970240639X,10.1016/j.dam.2007.04.017,10.1016/S0166-218X(99)00099-2}).

Given a positive Boolean function $f:\{0,1\}^n \to \{0,1\}$ expressed in its PIDNF, the \emph{dualization} problem \cite{doi:10.1137/S009753970240639X, doi:10.1137/S0097539793250299, 10.1016/j.dam.2007.04.017, journals/jal/FredmanK96,Mezzini20231}  consists in finding the PIDNF of a positive Boolean function $g$ such that $f(x)=\overline {g}(\overline {x})$ for all $x \in \{0,1\}^n$.

In \cite{journals/jal/FredmanK96}, two algorithms that we call Algorithm $A_1$ and Algorithm $A_2$ with complexity respectively of  $N^{O(\log^2 N)})$ and  $N^{o(\log N)}$ where $N=|F|+|G|$ and $|F|$ and $|G|$ are the number of terms of the PIDNF expression of $f$ and $g$ respectively.
In this paper we present a new algorithm for the dualization problem. Furthermore we present the results of extensive experiments comparing the running time of various algorithms. Surprisingly the new algorithm we present, whose complexity on general hypergraph is exponential, in practices, on hypergraph satisfying \eqref{eq:intersect}, are much faster than Algorithm $A_1$.

\section{Preliminaries and hypergraphs}\label{sec2}

In the following the variable $x$ is interpreted sometimes as a Boolean (or binary) $n$-dimensional vector and sometimes as the corresponding decimal expression of the binary vector. In particular if $x$ is the decimal value of the binary vector $(x_1, \dots, x_n)$ then the decimal value of the binary vector $(\overline x_1, \dots, \overline  x_n)$ is $\overline  x = 2^n-x-1$.
We recall here some well known propositions and lemmas about positive and/or self-dual Boolean functions.

\begin{proposition} [\cite{journals/jal/FredmanK96}] \label{prop:intersection}
Necessary condition for two positive Boolean functions $f=\bigvee_{e \in F} \bigwedge_{i \in e}x_i$ and $g=\bigvee_{t \in G} \bigwedge_{j \in t}x_j$ expressed in their PIDNF to be mutually dual is that 
\begin{equation} \label{eq:intersect}
e \cap t \neq \emptyset  \text{ for every $e \in F$ and $t \in G$}
\end{equation}
\end{proposition}

By Proposition \ref{prop:intersection}, if $f$ is self-dual then every implicant of $F$ must intersect every other implicant. So we have the following

\begin{lemma} [\cite{DBLP:journals/qmi/MezziniGGCPP24}] \label{less_than_2}
Let $f$ be a positive Boolean function which satisfies \eqref{eq:intersect}. Then $f(x)+ f(\overline{x}) \leq 1$ for all $0 \leq x < 2^n$.
\end{lemma}

\begin{lemma} [\cite{DBLP:journals/qmi/MezziniGGCPP24}]\label{balanced}
Suppose $f$ is self-dual. Then $f$ is balanced, that is, for $x$ in half of its domain $f(x)=0$  and on the other half of the domain $f(x)=1$.
\end{lemma}

\begin{lemma} [\cite{DBLP:journals/qmi/MezziniGGCPP24}] \label{self-dual}
Let $f$ be a positive Boolean function expressed in its PIDNF which satisfies also \eqref{eq:intersect}. Then $f$ is self-dual if and only if $\sum_{x=0}^{2^{n}-1} f(x)= 2^{n-1}$.
\end{lemma}

Choose $n>4$ odd and consider the following Boolean  function $f$ whose positive DNF expression $\Phi$ has as a set $F$ of implicants, the set of all subsets of $\{1,\dots, n\}$ of cardinality $\left \lceil n/2 \right \rceil$ where $\lceil a \rceil $ stands for the least integer greater or equal than $a$.
\begin {lemma} [\cite{DBLP:journals/qmi/MezziniGGCPP24}]  \label{number_of_implicants}
The function $f$ expressed by $\Phi$  is self-dual. Moreover $\Phi$ is the PIDNF representation of $f$, and has a number of terms equal to $\binom{n}{\left \lceil n/2 \right \rceil}$. 
\end{lemma}

\noindent

\begin{algorithm}
\caption{Algorithm Dual}
\label{alg:algoritm_dual}
\begin{algorithmic}[1]
\REQUIRE A PIDNF of a positive Boolean function $f$ satisfying \eqref{eq:intersect}.
\ENSURE \emph{TRUE} when $f$ is self-dual and \emph{FALSE} otherwise.
\STATE Let $S = \sum_{x=0}^{2^{n-1}-1} [f(x)+ f(\overline{x})]$
\IF {$S = 2^{n-1}$}
	\RETURN \emph{TRUE}
\ENDIF
\RETURN \emph{FALSE}
\end{algorithmic}
\end{algorithm}

The algorithm Dual is an application of Lemma \ref{self-dual}.
Now if $N$ is the number of terms of the PIDNF of $f$ and $n$ is the number of variables then the complexity of computing $f(x)$ is $O(nN)$. Furthermore the complexity of step 1 of algorithm Dual is $O(2^n)$. Therefore, when the number of terms $N$ is $O(2^n)$ as, for example, in the case of PIDNF $\Phi$ of Lemma \ref{number_of_implicants}, Algorithm Dual has complexity $O(nN^2)= O(N^2 \log N)$, which is polynomial in the dimension of the input. So we can state the following theorem.
\begin {theorem} [\cite{DBLP:journals/qmi/MezziniGGCPP24}] \label{polynomial-dual}
The asymptotic complexity of algorithm Dual is $O(N^{2} \log{N})$. 
\end{theorem}\vspace{0.3cm}
We finally observe that if $N = O(n^k)$ then Algorithm A has complexity $O(N^{o(\log (N)})=O(n^{k^2})$ which means that when $N$ is polynomially bounded by the number of variables then, by using Algorithm A, we have a polynomial time complexity for solving the dualization problem. Nevertheless, when $N=O(2^n)$ \emph{every} algorithm has complexity $\Omega(2^n)$, that is, no algorithm can be faster than $O(2^n)$.\\
Given all the above, we may state the following theorem which characterize the complexity of of the dualization problem.\\
\begin{theorem} [\cite{DBLP:journals/qmi/MezziniGGCPP24}]
The dualization problem has complexity:
\begin{itemize}
\item Polynomial in the number $n$ of variables of its PIDNF, when $N = O(n^k)$.
\item Polynomial in the number $N$ of terms of its PIDNF, when $N = O(2^n)$.
\end{itemize}
\end{theorem}

By using a classical quantum computer we may take advantage of the Grover search algorithm and we can speed up the complexity of the classical computer reaching complexity $O(N^{3/2} log N)$ \cite{DBLP:journals/qmi/MezziniGGCPP24}. Furthermore we can resort to a quantum annealer algorithm \cite{DBLP:journals/qmi/MezziniGGCPP24}.

\subsection{Hypergraphs}

Given a set $V$, a \emph{hypergraph} $H$ is a family of subsets of $V$, that is $H= \{e : e \subseteq V\}$. In the following we always consider hypergraphs $H$ for which $ \bigcup_{e \in H} e= V$ and denote $V= V(H)$. Each set $e \in H$  is called an \emph{hyperedge}.  A practical representation of a hypergraph is one in which the hypergraph is represented as a bipartite graph $G$ with bipartition $(V(H),H)$. An edge $\{i, e\}$ is in $G$ if and only if $i \in e$ where $i \in V(H)$ and $e \in H$. 
A hypergraph is \emph{connected} if and only if its bipartite representation is connected. There is a natural bijection between the connected components of hypergraph and the connected components of its bipartite representation.

If $p \subseteq e \in H$ we denote by $H - p$ the hypergraph obtained from $H$ by removing each vertex of $p$ from any hyperedge containing it. That is $H-p = \{e \setminus p | e \in H\}$.

\begin{example}
Consider, as an example, the  bipartite representation of the hypergraph $H= \{\{0,3\},\{0,4\},\{1,3,4\},\{0,1,2\},\{2,3,4\}\}$ of Figure \ref{fig:example}(A). Then $H-\{3\}$ is the hypergraph whose bipartite representation is in Figure \ref{fig:example}(B).
\end{example}
\vspace{0.3cm}

If $s \subseteq V(H)$ then we denote by $N_H(s)= \{e \in H|e \cap s \neq \emptyset\}$ and call it the \emph{neighbourhood} of $s$ in $H$. We also say that $s$ \emph{covers} or \emph{hits} $N_H(s)$.
A \emph{hitting set} of a hypergraph $H$ is a set $t \subseteq V(H)$ such that $N_H(t) = H$.
Denote by $hit_H$ the set of hitting sets of $H$, that is, $hit_H=\{t \subseteq V(H)|N_H(t)=H\}$. If $H = \emptyset$ then we set $|hit_H|=1$.

\section{New classical computer algorithms for solving the dualization problem}

Based on Lemma \ref{less_than_2} and Lemma \ref{self-dual} we may devise the following algorithms for the self-duality testing.  

\subsection{Counting all the hitting sets of a hypergraph}

If a positive Boolean function $f$ is not self-dual and its PIDNF satisfies \eqref{eq:intersect}, then, by Lemma \ref{self-dual}, $\sum_{x=0}^{2^{n}-1} (1-f(x))>2^{n}-\sum_{x=0}^{2^{n}-1} f(x)>2^{n}-2^{n-1}=2^{n-1}$. In other words if $S =\{x \in \{0,1\}^n| f(x)=0\}$ and $|S| > 2^{n-1}$ then $f$ is not self-dual, otherwise, by Lemma \ref{self-dual}, it is self-dual. If we consider $F$ as a hypergraph, then computing $|S|$ is equivalent to compute the number of hitting set of the hypergraph $F$.
In fact suppose that $t \in hit_F$ then let $x\in \{0,1\}^n$ such that $x_i = 0$ for $i \in t$ then, since $F$ satisfies \eqref{eq:intersect}, we have that $f(x)=0$. In order to do this we describe an algorithm for counting all the hitting sets of a hypergraph. \\\\
\noindent
Let $e$ be a hyperedge of $H$ and $s \subseteq e$, $s \neq \emptyset$. We define $hit_H(e,s)=\{t \in hit_H| t \cap e = s\}$.

\begin{lemma} \label{hit_subset}
The sets $hit_H(e,s)$ for $s \subseteq e$, $s \neq \emptyset$ form a partition of $hit_H$.
\end{lemma}
\begin{proof}
Given $s, p \subseteq e$, with $s \neq \emptyset$ and $p \neq \emptyset$, obviously we have that $hit_H(e,s) \cap hit_H(e,p)= \emptyset$ if $s \neq p$. 
Furthermore any $t \in hit_H$ must contain a non empty subset of $e$ and this proves that $\bigcup_{s \subseteq e, s \neq \emptyset} hit_H(e,s)= hit_H$. 
\end{proof}

\noindent
By Lemma \ref{hit_subset}, in order to compute the $|hit_H|$ we can compute $\sum_{s\subseteq e, s\neq \emptyset}|hit_H(e,s)|$. 

\begin{remark} \label{deletion}
Suppose that  $H$ is the PIDNF of a positive Boolean function satisfying \eqref{eq:intersect}. Let $e \in H$, $p \subseteq e$ and let $H_1=H - p$. Since in $H$ no hyperedge is subset of any other hyperedge, after the removal of the vertices in $p$, if $p \neq e$ then $|H_1|=|H|$.
\end{remark}
\vspace{0.3cm}

\noindent
Algorithm \ref{alg:hitting_sets_subset} computes $|hit_H(e,s)|$ for any $e \in H$ and any $s \subseteq e$, $s \neq \emptyset$. It first removes from $H$ all the vertices in $e \setminus s$ since these vertices are never used to build a hitting set $t $ such that $t \cap e = s$. After this, $s$ is an hyperedge of $H_1$. Next we obtain the hypergraph $H_2$ by removing from $H_1$ all the hyperedges in $N_{H_1}(s)$. If $n_1=|V(H1)\setminus (V(H_2) \cup s)| $ then the algorithm returns $2^{n_1} \cdot |hit_{H_2}|$.
\noindent

\begin{algorithm}
\caption{Counting Hitting Sets with Subset Removal}
\label{alg:hitting_sets_subset}
\begin{algorithmic}[1]
\REQUIRE A hypergraph $H\neq \emptyset$, a hyperedge $e \in H$, a non-empty $s \subseteq e$
\ENSURE The cardinality of  $hit_H(e,s)$
\STATE $H_1 \leftarrow H - (e \setminus s)$ 
\STATE $H_2 \leftarrow H_1 \setminus N_{H_1}(s)$
\STATE $n_1 \leftarrow |V(H_1) \setminus (V(H_2) \cup s)|$
\RETURN $2^{n_1} \cdot |hit_{H_2}|$
\end{algorithmic}
\end{algorithm}

\begin{lemma} \label{algorithm_subset}
The Algorithm \ref{alg:hitting_sets_subset} correctly computes $|hit_H(e,s)|$ for an $e \in H$ and an $s \subseteq e$, $s \neq \emptyset$.
\end{lemma}
\begin{proof}
We first note that any $t \in hit_H(e,s)$ does not contain any vertex in $e \setminus s$. Therefore, in step 1 of Algorithm \ref{alg:hitting_sets_subset}, we delete from $H$ all the vertices in $e \setminus s$ since they will not be used to obtain a hitting set $t$ such that $t \cap e=s$. Let $H_1= H - (e \setminus s)$. We note that, by Remark \ref{deletion}, $|H_1|=|H|$ since $s\neq \emptyset$. Therefore any hitting set of $H_1$ is a hitting set of $H$. Moreover $s \in H_1$, that is, $s$ is a hyperedge of $H_1$. 
At this point all we need to do is to count all the hitting sets of $H_1$ containing $s$. Let $H_2= H_1 \setminus N_{H_1}(s)$.  
We show now that every hitting set of $H_1$, containing $s$, can be obtained by the union of a hitting set $t_2$ of $H_2$ and any subset of $V(H_1)\setminus V(H_2)$ containing $s$. Let $t_2 \subseteq V(H_2)$ such that $N_{H_2}(t_2)= H_2$. Then clearly $s \cup t_2 \in hit_{H_1}$. On the other hand, let $t \in hit_{H_1}$ such that $s \subseteq t$ and let $t_1 = t \setminus s$.  Since, by definition, $N_{H_1}(s) \cap H_2=\emptyset$ we have that $H_2 \subseteq N_{H_1}(t_1)$. Let $t_2 = \{v | v  \in t_1 \cap f, f \in H_2\}$, that is $t_2$ is obtained by taking the union of the vertices in $t_1$ that are adjacent to some edge in $H_2$. Note that $s \cup t_2$ is a hitting set of $H_1$. Thus any $t\in hit_{H_1}$ is the union of $s$ with any hitting set $t_2 \subseteq V(H_2)$ of $H_2$ and any subset of $V_1=V(H_1)\setminus (V(H_2) \cup s)$. 
Therefore $|hit_H(e,s)|=2^{|V_1|}\cdot |hit_{H_2}|$.
\end{proof}

\begin{example} [continued] \label{ex:algorithm_hit}
Consider the  bipartite representation of the hypergraph of Figure \ref{fig:example}(A). Suppose we want to compute $hit_H(e, s)$ where $e =\{0,3\}$ and $s=\{0\}$. First we remove $v=3$ from the hypergraph by removing it from any hyperedge containing it. We eventually obtain the hypergraph $H_1$ of Figure \ref{fig:example}(B). Now $s$ is a hyperedge of $H_1$. After removing from $H_1$ the hyperedges in $N_{H_1}(s)$ we obtain the hypergraph $H_2$ of Figure \ref{fig:example}(C). It is immediate to see that $hit_{H_2}=\{\{4\},\{1,4\},\{2,4\},\{1,2\},\{1,2,4\}\}$. Since $V(H_1) \setminus (V(H_2) \cup s)=\emptyset$ then $n_1=0$ and $|hit_H(e, s)|= 2^{n_1} \cdot |hit_{H_2}|= 5$. In fact, it is not difficult to see that $hit_H(e, s)=\{\{0,4\},\{0,1,4\},\{0,2,4\},\{0,1,2\},\{0,1,2,4\}\}$.
\end{example}
\vspace{0.5cm}
\noindent
Now we are in position to describe the algorithm for counting all the hitting sets of a hypergraph. First note that if $H$ has $k$ connected component say $G_1, \dots, G_k$ then the number of hitting sets of the hypergraph is $\prod_{i=1}^k |hit_{G_i}|$. So the algorithm computes $|hit_{G_i}|$ for all connected components of $H$ by using Algorithm \ref{alg:hitting_sets_subset}.

\begin{figure}[ht]
    \centering
    \includegraphics[width=0.7\textwidth]{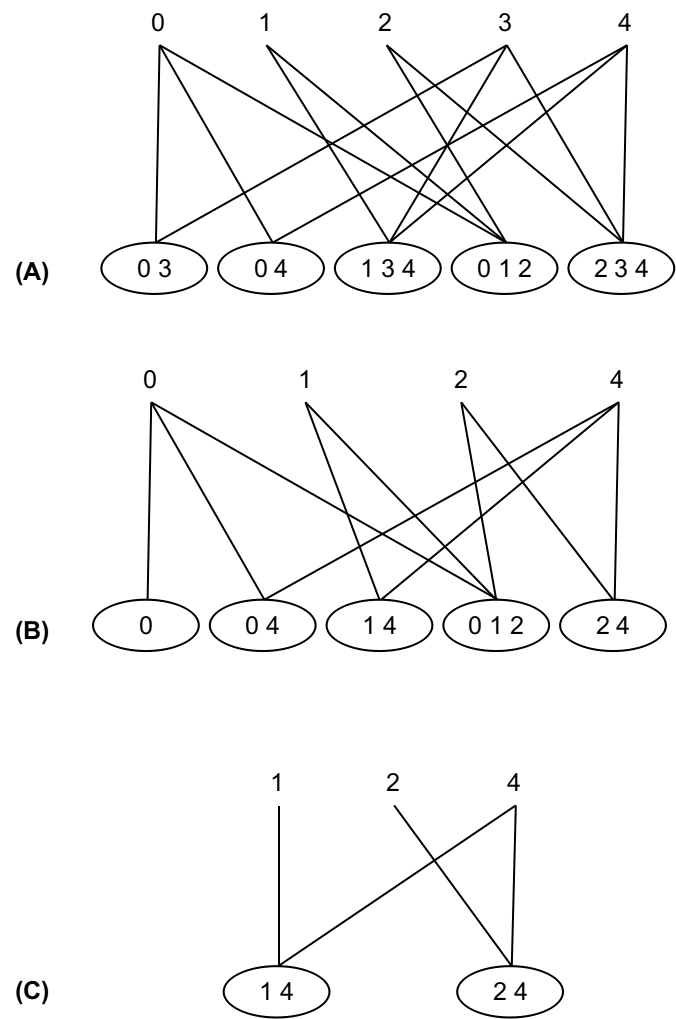}
    \centering        
    \caption{\textbf{(A)} The bipartite representation of the hypergraph $H= \{\{0,3\},\{0,4\},$ $ \{1,3,4\},\{0,1,2\},$ $\{2,3,4\}\}$. \textbf{(B)} The bipartite representation of the hypergraph $H_1= H- \{3\}$. \textbf{(C)} The bipartite representation of $H_2= H_1 \setminus N_{H_1}(0)$. There are five hitting sets in $H_2$: $\{ \{4\}, \{1,4\},\{2,4\}, \{1,2, 4\}, \{1,2\}\}$ 
    }
   \label{fig:example}
\end{figure}

\begin{algorithm}
\caption{Counting all the hitting sets in a hypergraph}
\label{alg:hitting_sets}
\begin{algorithmic}[1]
\REQUIRE The set of connected components $\{G_1, \dots, G_h\}$ of a hypergraph $H$
\ENSURE The cardinality of  $hit_H$
\STATE $nhit \leftarrow 1$
\FOR{$i =1$ to $h$}
    \STATE let $e$ be a hyperedge in $G_i$
	\STATE $nhit_i \leftarrow \sum_{s\subseteq e, s\neq \emptyset}|hit_H(e,s)|$
    \STATE $nhit \leftarrow nhit \cdot nhit_i$
\ENDFOR
\RETURN $nhit$
\end{algorithmic}
\end{algorithm}

\begin{lemma}
Algorithm \ref{alg:hitting_sets} correctly computes the number of hitting sets of a hypergraph.
\end{lemma}
\begin{proof}
The correctness directly follows from Lemma \ref{hit_subset} and  Lemma \ref{algorithm_subset}.
\end{proof}

\begin{example} [continued]
Consider the  bipartite representation of the hypergraph of Figure \ref{fig:example}(A) and let $e = \{0,3\}$. In  Example  \ref{ex:algorithm_hit} we already computed $|hit_H(e, \{0\})|$. Note that we need now to compute $|hit_H(e, \{3\})|$ and $|hit_H(e, \{0,3\})|$. In the first case it is easy to see that $H_2 = \{\{4\}, \{1,2\}\}$ and $H_2$ has two connected components. The first component, $\{\{4\}\}$, has only one hitting set and the second component $\{\{1,2\}\}$, has three hitting sets. Therefore $|hit_H(e, \{3\})|= 3$. 
As for $|hit_H(e, \{0,3\})|$ we note that $H_1=H$ and after removing $N_H(\{0,3\})$ from $H$ then $H_2= \emptyset$. In this case $|hit_{H_2}|=1$ and $V_1 = V(H_1)\setminus (V(H_2) \cup \{0,3\})= \{1,2,4\}$. Therefore $|hit_H(e, \{0,3\})|= 2^{|V_1|} \cdot hit_{H_2}= 2^3= 8$.
In the end, the sum of all hitting sets is $16$ and the hypergraph is self-dual, as it is easy to check. 
\end{example}
\subsection{Simple algorithm to search a minimal hitting set}

Given a PIDNF $H$ of a Boolean function $f$ satisfying \eqref{eq:intersect}, we can leverage Algorithm Dual in order to find, if any, a $x \in \{0,1\}^n$ such that $f(x)=f(\overline{x})$. By Lemma \ref{less_than_2}, if \eqref{eq:intersect} holds, then when $f(x)=1$ we are guaranteed that $f(\overline{x})=0$. However if \eqref{eq:intersect} holds and $f$ is not self-dual there must exists a $x \in \{0,1\}^n$ such that $f(x)=0$ and $f(\overline{x})=0$. If $x \in \{0,1\}^n$, we define the $w(x)$ the hamming weight of $x$, that is $w(x)=\sum_{i=0}^{n-1} x_i$. We search such an $x$ in Algorithm \ref{alg:minimal_hitting_sets}.

\begin{algorithm}
\caption{Search minimal hitting set}
\label{alg:minimal_hitting_sets}
\begin{algorithmic}[1]
\REQUIRE The PIDNF $F$ of a Boolean function satisfying \eqref{eq:intersect}.
\ENSURE True if $f$ is self-dual. 
\STATE $n \leftarrow |V(F)|$
\FOR{$i =1$ to $\left \lfloor \dfrac{n}{2} \right \rfloor$}
	\FOR{all $x \in \{0,1\}^n$ such that $w(x)=i$}
	    \IF {$f(x)=0$ and $f(\overline{x})=0$}
	      \RETURN $x$
    	\ENDIF
    \ENDFOR
\ENDFOR
\RETURN True
\end{algorithmic}
\end{algorithm}

\section{Experiments' results}

We generate hypergraphs with the following methodology. We pick a random number $x$ such that  $n_1\leq x < n_2$ where $0<n_1 < n_2< 2^n$. We choose $n_1= 2^{n-3}$ and $n_2 = 2^n-n_1$ where $n= |V(H)|$ is the number of variable of the Boolean function.  Then we consider the binary representation of $x$ as the characteristic vector of a set $t \subseteq \{0,1, \dots n-1\}$. We add $t$ to the hypergraph $H$ provided that no other hyperedge of $H$ is contained in $t$ and $t$ is not contained in any other hyperedge and that hypergraph satisfies \eqref{eq:intersect}.
Surprisingly Algorithm \ref{alg:minimal_hitting_sets} outperform all the other algorithms while, as expected, the algorithm for counting the hitting sets with brute force (\emph{HS brute force}) which counts all the hitting sets by checking for each $t \subseteq V(H)$ if $t$ is a hitting set, is the worst performing. It is also worth to note that Algorithm \ref{alg:hitting_sets}, surprisingly, outperform Algorithm A.

We made the software public on a colaboratory notebook
\footnote{\url{https://colab.research.google.com/drive/1CHGlgmKsk0pjJOubo_MqZZgw_2_Cq0oP?usp=sharing}}
. We run the tests on the same notebook and the execution times are reported in Table  \ref{tbl:execution_time}.

\begin{table}[h!]
    \centering
    \begin{tabular}{crrrrr}
        \toprule
        Vertices & Hyperedges & Algorithm A & Algorithm \ref{alg:hitting_sets} & HS Brute Force & Algorithm \ref{alg:minimal_hitting_sets} \\
        \midrule
        10 & 90 & 0.004 & 0.004 & 0.011 & 0.002 \\
        11 & 103 & 0.004 & 0.007 & 0.024 & 0.001 \\
        12 & 226 & 0.022 & 0.017 & 0.103 & 0.009 \\
        13 & 475 & 0.097 & 0.145 & 0.460 & 0.007 \\
        14 & 933 & 0.392 & 0.135 & 1.851 & 0.039 \\
        15 & 1655 & 1.325 & 0.340 & 7.091 & 0.101 \\
        16 & 2895 & 4.059 & 1.113 & 25.314 & 0.204 \\
        17 & 6477 & 20.945 & 4.886 & 116.119 & 1.953 \\
        18 & 11995 & 68.876 & 20.855 & 450.503 & 3.649 \\
        19 & 23573 & 274.425 & 68.606 & 1883.033 & 23.850 \\
        20 & 42271 & 971.873 & 328.349 & 6877.897 & 6.629 \\
        21 & 91937 & 5088.230 & 1435.641 & 32429.970 & 215.642 \\
        \bottomrule
    \end{tabular}
    \caption{Performance Comparison of Different Algorithms. The value reported for each algorithm are seconds of execution time. \label{tbl:execution_time}}
    \label{tab:algorithm_comparison_v2}
\end{table}

\section{Upcoming works}

We need to extends the experimentation to hypergraphs whose number of hyperedge is exponential to the number of vertices. We want to compare the running time of the algorithms with the execution time on a classical quantum computer and on a quantum annealer of the corresponding quantum (resp. quantum annealer) algorithms.
\clearpage

\bibliographystyle{abbrv}
\bibliography{sn-bibliography}

\end{document}